\RequirePackage{amsmath}
\documentclass[runningheads]{llncs}
\usepackage{graphicx}
\usepackage{amsfonts,amsmath}
\usepackage{xspace}
\usepackage{tabularx}
\usepackage{colortbl}
\usepackage{relsize}

\usepackage{algorithmic, algorithm}

\newcommand{\Oplus}{\mathlarger{\mathlarger{\mathlarger{\oplus}}}}

\newcommand{\pk}{\mathcmd{\mathit{pk}}} % public key
\newcommand{\sk}{\mathcmd{\mathit{sk}}} % secret key
\newcommand{\msg}{\mathcmd{m}} % message
\newcommand{\ctxt}{\mathcmd{c}} % ciphertext

\newcommand{\mathcmd}[1]{\ensuremath{#1}\xspace} % to use a command both in math mode and non-math mode
\newcommand{\Gen}{\mathcmd{\mathsf{Gen}}} % key generation algorithm
\newcommand{\Enc}{\mathcmd{\mathsf{Enc}}} % encryption algorithm
\newcommand{\Encpk}{\mathcmd{\mathsf{Enc}_{\pk}}}
\newcommand{\Dec}{\mathcmd{\mathsf{Dec}}} % decryption algorithm

\newcommand{\la}{\leftarrow}
\newcommand{\negl}{\mathcmd{\mathsf{negl}}} % negligable function
\newcommand{\PubK}{\mathcmd{\mathsf{PubK}}}
\newcommand{\CPA}{\mathcmd{\mathsf{CPA}}} % cpa
\newcommand{\adv}{\mathcmd{\mathcal{A}}}
\newcommand{\GenModulus}{\mathcmd{\mathsf{GenModulus}}}
\newcommand{\A}{\mathcmd{\mathtt{A}}}
\newcommand{\B}{\mathcmd{\mathtt{B}}}
\newcommand{\C}{\mathcmd{\mathtt{C}}}
\newcommand{\D}{\mathcmd{\mathtt{D}}}
\newcommand{\cvec}{\mathcmd{\mathbf{c}}}
\newcommand{\fvec}{\mathcmd{\mathbf{f}}}
\newcommand{\chisq}{\mathcmd{\chi^2}}
\newcommand{\chisqfc}{\mathcmd{\chisq(\fvec,\cvec)}}
\newcommand{\N}{\mathcmd{N}}
\newcommand{\PKE}{\mathcmd{\mathsf{PKE}}} % generic PKE name
\newcommand{\msgsp}{\mathcmd{\mathcal{M}}} % message space
\newcommand{\prot}{\Pi}
\newcommand{\func}{g}
\newcommand{\funcc}{\func_{c}}
\newcommand{\funcf}{\func_{f}}
\newcommand{\view}{\mathrm{view}}
\newcommand{\out}{\mathrm{output}}
\newcommand{\secp}{k}
\newcommand{\compind}{\cong}

\newcommand{\compactpar}[1]{\smallskip\noindent\textbf{#1}{\ }}
\newcommand{\remove}[1]{}
\newcommand{\Alice}{\text{Carol}\xspace}
\newcommand{\Bob}{\text{Felix}\xspace}
\newcommand{\carol}{\Alice}
\newcommand{\felix}{\Bob}
\newcommand{\simu}{\mathcal{S}}
\newcommand{\simup}[1]{\simu_{#1}}
\newcommand{\bits}{\{0,1\}}

\usepackage{booktabs, diagbox} % for tables

\usepackage{boxedminipage,paralist} % for minipage and figures

\begin{document}
\title{Secure Two-Party Feature Selection}

\author{Vanishree Rao\inst{1}\and 
Yunhui Long\inst{2}\and 
Hoda Eldardiry\inst{2}\and\\
Shantanu Rane\inst{2}\and 
Ryan A. Rossi\inst{3}\and 
Frank Torres\inst{2}}
\authorrunning{V. Rao et al.}

\institute{ Intertrust Technologies Corporation, Sunnyvale, CA 94085, USA; contribution to this work was done while this author was at PARC\\
\and
Palo Alto Research Center (PARC), Palo Alto CA 94304, USA\\
\and  
Adobe Research, San Jose CA 95110, USA\\
}

\maketitle             
\begin{abstract}
In this work, we study how to securely evaluate the value of trading data without requiring a trusted third party.
We focus on the important machine learning task of classification. 
This leads us to propose a provably secure four-round protocol that computes the value of the data to be traded without revealing the data to the potential acquirer.
The theoretical results demonstrate a number of important properties of the proposed protocol.
In particular, we prove the security of the proposed protocol in the honest-but-curious adversary model.

\end{abstract}

% !TEX root = main.tex

\section{Introduction}

According to the report ``Data Never Sleeps 6.0'' published recently by Domo Inc., an estimated 1.7 MB of data will be created every second for each person on earth by 2020. The owners of this staggering amount of data sometimes provide it readily to others, but often hold back despite the value that data trading could provide. Both privacy concerns and the desire to monetize data at a fair market value are barriers, as both could be compromised if data are revealed before terms have been negotiated. A method to assess the value of a data trade without first revealing the data would help make data trading a more efficient transaction, whether the aim is to trade at a fair market price, apply some type of differential privacy, or both.

\paragraph{Finding business value in `distributed' data:} When data on different aspects of a system are captured by different stakeholders, trading the data  can provide a more complete perspective of the system. For instance, in an Internet-of-Things (IoT) ecosystem, IoT devices owned by different parties (manufacturers, service providers, consumers, etc.) often collect data that reveal only a partial understanding of behaviors and events. Creating a marketplace for trading the data would enable a party to get a more complete understanding when required, without spending extra time and money deploying additional IoT devices to collect data that another party already has. As long as stakeholders can establish a fair price for the data, inefficient duplication of efforts can be avoided, benefiting both parties of a transaction. However, identifying trade partners and tagging a cash value to the data can be a tricky challenge, particularly because the value depends on the quality and content of the data held by both partners.

\paragraph{Maximizing data utility while protecting individual privacy:} When considering how to share sensitive datasets, potential collaborators may seek to analyze how different statistical privacy options affect the utility of data.  The party applying statistical privacy to their data before sharing may like to work with a potential collaborator to experiment with different choices of statistical privacy methods and parameters, in order to deliver desensitized data of the highest possible utility. Applications include both business-to-business transactions and business-to-government transactions.

\paragraph{Data trading scenarios:} An owner of a dataset may want to release only subsets of their data to control proliferation, but they need a way to determine utility of subsets in order to choose the right one for each potential collaborator.  An owner may also want to limit the number of times data are shared, either to mitigate security and privacy risks or to maintain a desired monetary price for access to the data.  Choosing customers that have the highest utility for the data will help maximize monetary return, as those customers will in principle pay a higher price.  An owner may want to sell access to data at a full value-based price, but rational purchasers may insist on a discounted price to compensate for any risk associated with uncertain utility.  
Thus, answering the following question is important:

\vspace*{\fill} 
\begin{quote} 
\centering 
How can one securely measure utility of data and the impact of applying statistical privacy enhancement techniques, without access to the actual data?
\end{quote}
\vspace*{\fill}

\subsection{This Work}

In this work, we try to answer the above question for a specific potential acquirer's task, where the parties freely share data dictionaries. Specifically, we provide a protocol with which a potential provider and a potential acquirer can determine the value of the data with respect to the latter's task at hand, without the latter learning anything more about the data, other than its specification in the data dictionary.
The specific sub-case we consider is the provider having a binary feature vector and the acquirer having a binary class vector. The acquirer would like to learn if the provider's feature vector can improve the correctness of the acquirer's classification. Thus, the utility we consider is whether the data shared by the provider is expected to improve the classification of the  acquirer's existing dataset.  To quantify utility, we use the \chisq-statistic studied by Yang and Pederson (1997) for the related problem of feature selection. We employ Pallier homomorphic encryption for the required privacy-preserving computations.

\subsection{Roadmap}

The protocols in this paper assume parties share primary keys for their data, in order for data elements to be aligned. In future work, we will integrate private set intersection protocols, such as the Practical Private Set Intersection Protocols published by De Cristofaro and Tsudik~\cite{de2010practical}, in order to relax this assumption. We also plan to study extensions of the work to more sophisticated feature selection, based on combining multiple columns in the provider's dataset to generate more complex feature candidates.  

% !TEX root = main.tex

\section{Background} \label{sec:background}

In this work, we consider a structured dataset, and we are interested in  classification based on all the features available. Specifically, we consider two parties, \carol and \felix. \carol has  a dataset consisting of certain feature columns and a class vector generated from her available features. \felix possesses an additional feature column $\mathbf{f}$ that might be useful for \carol in improving the classification of her dataset. 

\compactpar{Notations.}
Let $\mathbf{c} = (c_1, \mathrm{c_2}, \dots, c_n)$ be the class label vector with \carol, and $\mathbf{f} = (\mathrm{f_1}, \mathrm{f_2}, \dots, \mathrm{f_n})$ be the feature vector with \felix. We assume both the class labels and the features are binary attributes, leaving generalization to multinomial classifiers for a future paper. That is, for all $1 \le i \le n$, $\mathrm{c_i} \in \{0,1\}$ and $\mathrm{f_i} \in \{0,1\}$. Let $\mathrm{c_i}$ denote the class variable of the $i$-th record in \carol's dataset. Let  $\mathrm{f_i}$ be the feature value, in \felix's feature vector, corresponding to the $i$-th record in \carol's dataset. \remove{That is, $\mathbf{c}$ and $\mathbf{f}$ can be viewed as resulting from vertically partitioning a database with $n$ rows, where the $i^{th}$ row has attributes $(\mathrm{f_i}, \mathrm{c_i})$ for all $1 \le i \le  n$.}

\subsection{\chisq Feature Selection}
Feature selection is the process of removing non-informative features and selecting a subset of features that are useful to build a good predictor~\cite{guyon2003introduction}. The criteria for feature selection vary among applications. For example, Pearson correlation coefficients are often used to detect dependencies in linear regressions, and mutual information and \chisq statistics are commonly used to rank discrete or nominal features~\cite{guyon2003introduction,yang1997comparative}. 

In this paper, we focus on determining utility of binary features. We choose \chisq statistics as a measure of utility, due to its wide applicability and its amenability towards cryptographic tools. More specifically, unlike mutual information which involves logarithmic computations, the calculation of \chisq statistics only involves additions and multiplications.  

For the class label vector $\mathbf{c}$ and the corresponding feature vector $\mathbf{f}$, \A is defined to be the number of rows with $\mathrm{f_i} = 0$ and $\mathrm{c_i} = 0$. \B is defined to be the number of rows with $\mathrm{f_i} = 0$ and $\mathrm{c_i} = 1$. \C is defined to be the number of rows with $\mathrm{f_i} = 1$ and $\mathrm{c_i} = 0$. \D is defined to be the number of rows with $\mathrm{f_i} = 1$ and $\mathrm{c_i} = 1$.  Table~\ref{table:ctable} shows the two-way contingency table for \fvec and \cvec. The \chisq statistic of \fvec and \cvec is defined~\cite{yang1997comparative} to be:
\begin{equation*}
\chisq(\fvec, \cvec) = \frac{n(\A\D - \B\C)^2}{(\A+\C)(\A+\B)(\C+\D)(\B+\D)}.
\end{equation*}

\begin{table}[h]
\caption{Two-Way Contingency Table of \fvec and \cvec}
\label{table:ctable}
\centering
\begin{tabularx}{.5\textwidth}{| l | X | X |}
\hline
\diagbox[width=.2\textwidth]{\fvec}{\cvec} & 0 & 1 \\
\hline
0 & \A & \B  \\
\hline
1 &  \C & \D \\
\hline
\end{tabularx}
\end{table}

\noindent
$\chisq(\fvec, \cvec)$ is used to test the independence of \fvec and \cvec. Table~\ref{table:chi-table} shows the confidence of rejecting the independence hypothesis under different \chisq values. For example, when $\chisq(\fvec, \cvec)$ is larger than $10.83$, the independence hypothesis can be rejected with more than 99.9\% confidence, indicating that the feature vector \fvec is very likely to be correlated with the class label vector \cvec. 

\begin{table}[h]
\caption{Confidence of Rejecting the Hypothesis of Independence under Different \chisq Values}
\label{table:chi-table}
\centering
\begin{tabularx}{0.4\textwidth}{ c X  X }
\toprule
%\hline
& \chisq(\fvec, \cvec) & \textbf{Confidence} \\
\midrule
%\hline
& 10.83 & 99.9\%  \\
%\hline
& 7.88  & 99.5\% \\
%\hline
& 6.63 & 99\% \\
%\hline
& 3.84 & 95\% \\
%\hline
& 2.71 & 90\% \\
\bottomrule
\end{tabularx}
\end{table}

\subsection{Cryptographic Tools} \label{subsec:crypto_tools}
\subsubsection{PKE scheme and CPA security.}
We recall the standard definitions of public-key encryption (PKE) schemes and chosen plaintext attack (CPA) security, which are used in this paper. 

\compactpar{PKE schemes.} A  scheme \PKE with
message space \msgsp consists of three probabilistically-polynomial-time (PPT)  algorithms \(\Gen,\Enc,\Dec\).
Key generation algorithm \(\Gen(1^k)\) outputs a public key \pk and a secret key \sk.
Encryption algorithm \(\Enc(\pk,\msg)\) takes \pk and a message
\(\msg\in\msgsp\), and outputs a ciphertext \ctxt.  Decryption algorithm
\(\Dec(\sk,\ctxt)\) takes \sk and a ciphertext \ctxt, and outputs
a message \msg. For correctness, we require that \(\Dec(\sk,\ctxt)=\msg\) for all
\(\msg\in\msgsp\), all \((\pk,\sk)\gets\Gen(1^k)\), and all
\(\ctxt\gets\Enc(\pk,\msg)\).

\remove{\begin{definition}[PKE Scheme~\cite{katz2014introduction}.]
A \emph{public-key encryption (PKE) scheme} is a triple of probabilistic polynomial-time algorithms (\Gen, \Enc, \Dec) such that: 
\begin{compactitem}
\item The \emph{key-generation algorithm} $\Gen(1^k)$ outputs a pair of keys $(\pk,\sk)$. 
\item The \emph{encryption algorithm} $\Enc_{\pk}(m)$ takes as input a public key \pk and a message $m$ from some message space, and outputs ciphertext $c$. 
\item The \emph{deterministic decryption algorithm} $\Dec_{\sk}(c)$ takes as input a private key \sk and a ciphertext $c$, and outputs a message $m$ or a special symbol $\bot$ denoting failure.
\end{compactitem}
It is required that, except possibly with negligible probability over $(\pk, \sk)$ output by $\Gen(1^k)$, we have $\Dec_{\sk}\left(\Enc_{\pk}(m)\right) = m$ for any (legal) message $m$.
\end{definition}
}

\compactpar{Negligible Function.}
A function $f : \mathbb{N} \rightarrow \mathbb{R}$ is \emph{negligible} if for every possible integer $c$, there exists an integer $N$ such that for all $x > N$, 
	$| f(x) | \le \frac{1}{x^c}$.
We denote negligible functions as $\negl(\cdot)$.

\compactpar{The CPA Experiment.}
We now describe the chosen-plaintext attack (CPA) game  with an adversary $\adv$ against a PKE scheme $\PKE$.

\begin{algorithm} 
\caption{The $\PubK^{\CPA}_{\adv, \PKE}$ Experiment}
\label{algo:pubk}
\begin{algorithmic}[1]
	\REQUIRE	Security parameter $k$
    \STATE $(\pk,\sk) \gets \Gen(1^k)$
    \STATE The adversary \adv is given $1^k$, $\pk$, and oracle access to $\Enc_{\pk}(\cdot)$. \adv outputs a pair of messages $(m_0, m_1)$ of the same length
    \STATE A uniform bit $b \in \bits$ is chosen, and $c \gets \Enc_{\pk}(m_b)$ is given to \adv
    \STATE \adv continues to have access to $\Enc_{\pk}(\cdot)$, and outputs a bit $b'$
    \ENSURE $1$ if $b' = b$, and $0$ otherwise
\end{algorithmic}
\end{algorithm}

\compactpar{CPA Security~\cite{katz2014introduction}.}
A PKE scheme $\PKE = (\Gen, \Enc, \Dec)$ has indistinguishable encryptions under a \emph{chosen-plaintext attack}, or is \emph{CPA-secure}, if for all probabilistic polynomial-time adversaries \adv there is a negligible function \negl such that
\begin{equation*}
\Pr\left[ \PubK^{\CPA}_{\adv, \PKE} \left( k \right) = 1 \right] \le \frac{1}{2} + \negl(k),
\end{equation*}
where the experiment $\PubK^{\CPA}_{\adv, \PKE}$ is defined in Algorithm~\ref{algo:pubk}, and the probability is taken over the randomness of \adv and of the experiment.

\subsubsection{Paillier Encryption.}
We use Paillier encryption to maintain privacy in our two-party feature selection algorithm, and employ the additive homomorphic property of Paillier encryption to calculate the \chisq statistics that quantify feature utility. We recall the Paillier encryption scheme in Figure~\ref{fig:paillier}~\cite{katz2014introduction}.

Note that while we use Paillier homomorphic encryption, the proposed protocols can accomodate any semantically secure additively homomorphic encryption scheme.

\begin{figure*}[h!]
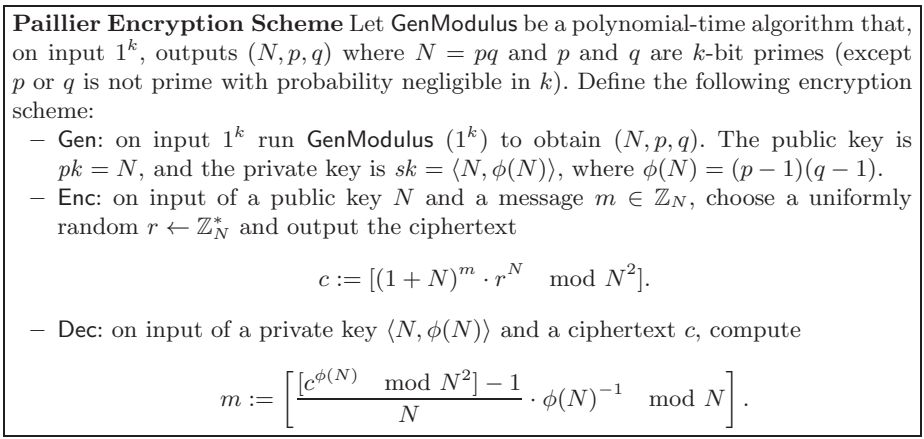

\begin{boxedminipage}{\textwidth}
\textbf{Paillier Encryption Scheme}
Let \GenModulus be a polynomial-time algorithm that, on input $1^k$, outputs $(N, p, q)$ where $N = pq$ and $p$ and $q$ are $k$-bit primes (except $p$ or $q$  is not prime  with probability negligible in $k$). Define the following encryption scheme:
\begin{compactitem}
\item \Gen: on input $1^k$ run \GenModulus($1^k$) to obtain $(N, p, q)$. The public key is $\pk = N$, and the private key is $\sk = \langle N, \phi(N) \rangle$, where $\phi(N) = (p-1)(q-1)$.
\item \Enc: on input of a public key $N$ and a message $m \in \mathbb{Z}_N$, choose a uniformly random $r \gets \mathbb{Z}^*_N$ and output the ciphertext
\begin{displaymath}
c := [(1+N)^m \cdot r^N\mod N^2].
\end{displaymath}
\item \Dec: on input of a private key $\langle N, \phi(N) \rangle$ and a ciphertext $c$, compute 
\begin{displaymath}
m := \left[\frac{[c^{\phi(N)}\mod N^2] -1}{N} \cdot \phi(N)^{-1}\mod N\right].
\end{displaymath}
\end{compactitem}
\end{boxedminipage}
  \caption{\label{fig:paillier} Paillier Encryption Scheme.}
\end{figure*}

Paillier encryption supports additive and scalar multiplication homomorphism. 
We briefly recall the definitions of additive homomorphism and scalar multiplication homomorphism~\cite{katz2014introduction}.

\compactpar{Additive Homomorphism.} A PKE scheme \PKE = (\Gen, \Enc, \Dec) is said to be \emph{additively homomorphic}, if there exists a binary operation $\oplus$, such that the following holds for all $k \in \N$, and for all $m_1, m_2 \in \msgsp$,
\begin{equation*}
\begin{aligned}
\Pr \left[ m^* = m_1 + m_2 \left\vert
\begin{array}{ll}
(\pk, \sk) \gets \Gen(1^k) \\
\mathrm{c_1} \gets \Encpk(m_1), \mathrm{c_2} \gets \Encpk(m_2) \\
c^* \gets \mathrm{c_1} \oplus \mathrm{c_2} \\
m^* \gets \Dec_{\sk}(c^*) 
\end{array} 
\right.
\right] = 1 - \negl(k).
\end{aligned}
\end{equation*}

\compactpar{Scalar Multiplication Homomorphism.} A PKE scheme \PKE = (\Gen, \Enc, \Dec) is said to be \emph{scalar multiplication homomorphic}, if there exists a binary operation $\otimes$, such that the following holds for all $k \in \N$, and for all $m_1, m_2 \in \msgsp$, 
\begin{equation*}
\begin{aligned}
\Pr \left[ m^* = m_1 m_2 \left\vert
\begin{array}{ll}
(\pk, \sk) \gets \Gen(1^k) \\
c \gets \Encpk(m_2) \\
c^* \gets m_1 \otimes c \\
m^* \gets \Dec_{\sk}(c^*) 
\end{array} 
\right.
\right] = 1 - \negl(k).
\end{aligned}
\end{equation*}

% !TEX root = main.tex

\section{Proof of Privacy}

We first present  the high-level argument for how our protocols will protect each party's data. We have one of the parties (\carol) choose the encryption key, and encrypt  her data using this key before sending it to the other party (\felix). Thus, \carol's privacy will be guaranteed by the semantic security assumption of the encryption scheme. Meanwhile, \felix will also encrypt his data using \carol's key, but he will blind all of the outputs he sends to \carol with randomness of his choosing, ensuring that \carol can learn nothing about his data.
We now make these notions precise by first providing a formal definition of privacy protection in the honest-but-curious adversary model, and a formal proof of privacy for the  protocol that attempts to protect privacy in the above described manner. 

\begin{definition}[Honest-but-curious security of two-party protocol]
\label{def:secdef}
We begin with the following notation:
\begin{itemize}
\item Let $\funcc$ and $\funcf$ be  probabilistic polynomial-time functionalities and let $\prot$ be a two-party protocol for computing $\func = (\funcc, \funcf)$. Let the parties be $\carol, \felix$, with inputs \cvec,\fvec respectively. 
\item The $\view$ of the  party $A \in \{\carol, \felix\}$ during an execution of $\prot$ on $(\cvec,\fvec)$ and security parameter $\secp$ is denoted by $\view^{\prot}_{A}(\cvec,\fvec,\secp)$ and equals $(w, r^A, m^A_1,\ldots,m^A_t)$, where $w \in \{\cvec, \fvec\}$  ($w$'s value depending on the value of $A$), $r^A$ equals the contents of the  party $A$'s internal random tape, and $ m^A_j$ represents the $j$-th message that it received.
\item The output of the  party $A$ during an execution of $\prot$ on $(\cvec,\fvec)$ and security parameter $\secp$ is denoted by  $\out^{\prot}_{A}(\cvec,\fvec,\secp)$  and can be computed from its own view of the execution. 
\end{itemize}

Let  $\func = (\funcc, \funcf)$ be a functionality. We say that $\prot$ securely computes $\func$ in the presence of semi-honest adversaries if there exist probabilistic polynomial-time algorithms $\simup{c}$ and $\simup{f}$ such that
\begin{eqnarray}
& \{(\simup{c}(1^\secp, \mathbf{c}, \funcc(\mathbf{c},\mathbf{f})), \func(\mathbf{c},\mathbf{f}))\}_{\mathbf{c},\mathbf{f},\secp}  \compind   \{(\view^{\prot}_{\carol}(\mathbf{c},\mathbf{f},\secp), \out^{\prot}(\mathbf{c},\mathbf{f},\secp))\}_{\mathbf{c},\mathbf{f},\secp}  \label{eqn:simuone} \\
& \{(\simup{f}(1^\secp, \mathbf{f}, \funcf(\mathbf{c},\mathbf{f})), \func(\mathbf{c},\mathbf{f}))\}_{\mathbf{c},\mathbf{f},\secp} \compind   \{(\view^{\prot}_{\felix}(\mathbf{c},\mathbf{f},\secp), \out^{\prot}(\mathbf{c},\mathbf{f},\secp))\}_{\mathbf{c},\mathbf{f},\secp} \label{eqn:simutwo}
\end{eqnarray}
$\mathbf{c},\mathbf{f }\in \bits^*$ such that $|\mathbf{c}| = |\mathbf{f}|$, and $\secp \in \mathbb{N}$.

%  \label{eqn:simuone}  \label{eqn:simutwo}

\remove{
Suppose that protocol $\prot$ has \carol compute (and output) the function $\funcc(c,f)$, and has \felix compute (and output) $\funcf(c,f)$, where $c,f$ denotes the inputs for \carol and \felix (respectively). Let VIEWA(x,y) (resp. VIEWB(x,y)) represent \carol's (resp. \felix's) view of the transcript. In other words, if (x,rA) (resp. (y,rB)) denotes \carol's (resp. \felix's) input and randomness, then:
where the {mi} denote the messages passed between the parties. Also let OA(x,y) and OB(x,y) denote \carol's (resp. \felix's) output. Then we say that protocol X protects privacy (or is secure) against an honest-but-curious adversary if there exist probabilistic polynomial time simulators \simup{c} and \simup{f} such that:
}

\end{definition}

% !TEX root = main.tex

\section{Protocol}
\label{sec:protocol}
In this section, we describe a \emph{four}-round protocol for \chisq statistic calculation under a two-party setting. For convenience, we continue to refer to the  parties as  \emph{\carol}, who has the class vector \cvec,  and \emph{\felix}, who has the feature vector \fvec. \carol's objective is to learn \chisqfc and \felix's objective is to not reveal any further information about \fvec while Carol computes the utility of Felix's data for her classifier. In this section, Felix uses multiplicative binding to keep the detailed mathematics a little simpler, but an alternative protocol that uses additive blinding is provided in Section 6 for situations where the security of multiplicative blinding is a concern.

As before, \A is the number of rows with $\mathrm{f_i} = 0$ and $\mathrm{c_i} = 0$. \B is the number of rows with $\mathrm{f_i} = 0$ and $\mathrm{c_i} = 1$. \C is the number of rows with $\mathrm{f_i} = 1$ and $\mathrm{c_i} = 0$. \D is the number of rows with $\mathrm{f_i} = 1$ and $\mathrm{c_i} = 1$.

\medskip
\noindent{\bf{Round 1}.} \\
\noindent \carol performs the following operations:
\begin{compactenum}
\item Generate a Paillier key pair $(\pk, \sk) = \Gen(1^k)$.
\item Encrypt all class labels with $\pk$: $\Encpk(\mathrm{c_1}),  \Encpk(\mathrm{c_2}),\ldots, \Encpk(\mathrm{c_n})$.
\item Compute $\frac{\B+\D}{\A+\C}$. Note that \carol can obtain this value by computing  $\frac{\sum_{i=1}^{n}\mathrm{c_i}}{n - (\sum_{i=1}^{n}\mathrm{c_i})}$, since $\B+\D = \sum_{i=1}^{n}\mathrm{c_i}$ and $\A+\C = n - (\B+\D)$, based on the contingency table.
\item Encrypt $\frac{\B+\D}{\A+\C}$ with $\pk$: $\Encpk\left(\frac{\B+\D}{\A+\C}\right)$.
\item Send the following  values to \felix: 
\begin{equation*}
\left(\pk, \Encpk(\mathrm{c_1}), \Encpk(\mathrm{c_2}), \dots, \Encpk(\mathrm{c_n}), \Encpk\left(\frac{\B+\D}{\A+\C}\right)\right). 
\end{equation*}
\end{compactenum}

\noindent{\bf{Round 2}.} \\
\noindent \felix performs the following operations:
\begin{compactenum}
\item Compute $\Encpk(\D)$. Note that Felix can obtain this value by computing $\Oplus_{i=1}^{n} \left( \mathrm{f_i} \otimes \Encpk(\mathrm{c_i}) \right) =   \Oplus_{i=1}^{n}  \Encpk(\mathrm{f_i}\mathrm{c_i})$ = $ \Encpk(\sum_{i=1}^{n}\mathrm{f_i}\mathrm{c_i})$, 

since $\sum_{i=1}^{n}\mathrm{f_i}\mathrm{c_i} = \D$.
\item  Sample  $r \la \mathbb{Z}_N$, and compute 
$r\otimes\Encpk(\D) = \Encpk(r\D)$.
\item Send the following  value to \carol:
\begin{equation*}
\Encpk(r\D)
 \end{equation*}
\end{compactenum}

\noindent{\bf{Round 3}.} \\
\noindent \carol performs the following operations:
\begin{compactenum}
\item Decrypt $\Encpk(r\D)$ using $\sk$.
\item Compute $\frac{r^2 \D^2}{(\B+\D)(\A+\C)}$ and $\frac{r\D}{\A+\C}$, and encrypt them.
\item Send the following values to \felix:
\begin{equation*}
\left(\Encpk\left(\frac{r^2 \D^2}{(\B+\D)(\A+\C)}\right), \Encpk\left(\frac{r\D}{\A+\C}\right)\right).
 \end{equation*}
\end{compactenum}

\noindent{\bf{Round 4}.} \\
\noindent \felix performs the following operations:
\begin{compactenum}

\item Cancel $r$ by computing
\begin{equation*}
{r^{-2}} \otimes \Encpk\left(\frac{r^2 \D^2}{(\B+\D)(\A+\C)}\right) = \Encpk\left(\frac{\D^2}{(\B+\D)(\A+\C)}\right) 
\end{equation*}
and
\begin{equation*}
{r^{-1}} \otimes \Encpk\left(\frac{r\D}{\A+\C}\right) = \Encpk\left(\frac{\D}{\A+\C}\right).
\end{equation*}
\item Compute an encryption of $\chisqfc$ by computing:
\begin{equation*}
\begin{aligned}
& \left( \frac{n^3}{(\A+\B)(\C+\D)} \otimes \Encpk \left( \frac{\D^2}{(\B+\D)(\A+\C)} \right)  \right)\\
&\oplus \left(\frac{n(\C+\D)}{\A+\B} \otimes \Encpk\left(\frac{\B+\D}{\A+\C}\right)\right)
\oplus \left(\frac{-2n^2}{\A+\B} \otimes \Encpk \left(\frac{\D}{\A+\C}\right) \right),
\end{aligned}
\end{equation*}
where $\C+\D $ and $\A+\B $ are computed as
\begin{equation*}
\C+\D = \sum_{i=1}^{n} \mathrm{f_i},
\end{equation*}
and
\begin{equation*}
\A+\B = n - (\C+\D).
\end{equation*}
We see below that the above computation gives $\Encpk(\chisqfc)$. 
Since $\A\D -\B\C = (\A+\B+\C+\D)\D - (\B+\D)(\C+\D)$, \chisqfc can be decomposed as follows:
\begin{equation*}
\begin{aligned}
\chisqfc &= \frac{n(\A\D-\B\C)^2}{(\A+\C)(\A+\B)(\C+\D)(\B+\D)} \\
&= \frac{n^3}{(\A+\B)(\C+\D)} \frac{\D^2}{(\B+\D)(\A+\C)} + \frac{n(\C+\D)}{(\A+\B)} \frac{(\B+\D)}{(\A+\C)} - \frac{2n^2}{(\A+\B)} \frac{\D}{(\A+\C)}. \\
\end{aligned}
\end{equation*}
\item Send the following value  to \carol:
\begin{equation*}
\Encpk(\chisqfc).
 \end{equation*}
\end{compactenum}

\noindent{\bf{Local computation}.} \\
\carol decrypts $\Encpk(\chisqfc)$ to obtain \chisqfc.

\remark{We note that only \carol receives the value \chisqfc. Depending on the application, if \felix also needs to know the value of \chisqfc, \carol can simply then send  it to \felix after running the protocol.
}
\remark{If Felix needs to know the value of \chisqfc but does not trust Carol to send the true value, then the parties can use a two-stream protocol wherein both parties compute and send encrypted values in round one and both parties send encrypted values of \chisqfc in round four.  Since the computation for \chisqfc is symmetric with respect to c and f, both parties should end up with the same value of \chisqfc, assuming they used the same data in both streams (i.e., did not cheat).  To verify that the parties did not cheat, they can re-encrypt their \chisqfc values with a new, single-use key, send their re-encrypted \chisqfc to the other party, and then send the one-use key after receiving the re-encrypted \chisqfc message from the other party.  If cheating occurred, the decrypted value of the other party's \chisqfc will not match their own.   
}

% !TEX root = main.tex
\section{Proof of Security}

With respect to the notion of security specified in Definition~\ref{def:secdef}, we first prove the following key lemma that will allow us to argue that our two-party protocol is secure against an honest-but-curious adversary. Specifically, the lemma captures the crux of proof, and its extension to the main theorem is straightforward.
\begin{lemma}
\label{lemma:core}
Suppose that in a two-party protocol $\prot'$, \carol runs the key generation algorithm of a CPA-secure homomorphic public-key encryption scheme and gives the public key to \felix. Also, suppose that all messages sent from \carol to \felix are encrypted with the generated public key, and all messages sent from \felix to \carol are either encryptions of elements randomly distributed in the plaintext space  and  independent of \felix's inputs, or encryptions of the final output. Then, the protocol $\prot$ is secure in the honest-but-curious adversary model.
\end{lemma}

\begin{proof}
To prove the security of the protocol, we need to consider two cases -- one, where \carol is corrupted, and the other, where \felix is corrupted. In each case, we will prove that the corrupted party will not learn anything more about the other party's output than the protocol output. Specifically, we show that there exist PPT algorithms $\simup{c}$ and $\simup{f}$, that simulate the non-corrupted party's messages without knowing the non-corrupted party's inputs but only knowing the output, in cases where \carol and \felix are corrupted, respectively. This corresponds to establishing equations (\ref{eqn:simuone}) and (\ref{eqn:simutwo}) in Definition~\ref{def:secdef}. 

\noindent \emph{Case 1: When \felix is corrupted by an adversary.} We show how to simulate \carol's messages sent to \felix, by describing the simulator $\simup{f}$. For every ciphertext to be sent from \carol to \felix, $\simup{f}$ chooses a random plaintext in the message space and sends an encryption of it. If \felix can tell apart the views of communicating with \carol and with the simulator, then there exists an adversary that can break CPA security of the underlying encryption scheme. Since, by assumption, no such PPT adversary exists, we have that Equation (\ref{eqn:simutwo}) holds. 

\noindent \emph{Case 2: When \carol is corrupted by an adversary.} We show how to simulate \felix's messages sent to \carol, by describing the simulator $\simup{c}$. For every ciphertext that encrypts a randomly distributed plaintext, sent by \felix to \carol, $\simup{c}$ samples a uniform random element in the plaintext space, encrypts it with \carol's public key, and sends the resulting ciphertext to \carol. For the ciphertext encrypting the final output, note that  $\simup{c}$ gets the final output as an input. Using this, the simulator can compute its encryption, and send the resulting ciphertext to \carol. Since the messages sent by $\simup{c}$ to \carol are distributed identically to \felix's messages  to \carol, we have that Equation (\ref{eqn:simuone}) holds. 
\end{proof}

We will now simply extend the core lemma into the main theorem. 

\begin{theorem}
The two-party protocol $\prot$ described in Section~\ref{sec:protocol} is secure in the honest-but-curious adversarial model.
\end{theorem}

\begin{proof}
We note the following aspects in the  protocol $\prot$. All the messages sent from \carol to \felix are encrypted using \carol's public key under Paillier encryption scheme. The messages sent from \felix to \carol are either encryptions of elements randomly distributed in $ \mathbb{Z}_N$, the plaintext space, or encryption of the final output, $\chisqfc$. Since these aspects conform to the conditions in Lemma~\ref{lemma:core}, based on the lemma, we have that the protocol $\prot$ is  secure in the honest-but-curious adversary model.
\end{proof}

% !TEX root = main.tex

\section{Alternative Protocol}
\label{sec:altprotocol}
In this section, we describe an alternative protocol for \chisq statistic calculation under a two-party setting, wherein Felix uses additive blinding rather than multiplicative blinding to introduce the random number $r$.  In theory, taking advantage of additive rather than multiplicative homomorphism provides stronger security~\cite{bianchi2011analysis}, albeit at a cost in computational efficiency and complexity. For this alternative protocol, round one is unchanged:

\medskip
\noindent{\bf{Round 1}.} \\
\noindent \carol performs the same operations as in Round 1 of section 4, including sending the following  values to \felix: 
\begin{equation*}
\left(\pk, \Encpk(\mathrm{c_1}), \Encpk(\mathrm{c_2}), \dots, \Encpk(\mathrm{c_n}), \Encpk\left(\frac{\B+\D}{\A+\C}\right)\right). 
\end{equation*}

\noindent{\bf{Round 2}.} \\
\noindent \felix performs the following operations:
\begin{compactenum}
\item Compute $\Encpk(\D)$. Note that Felix can obtain this value by computing $\Oplus_{i=1}^{n} \left( \mathrm{f_i} \otimes \Encpk(\mathrm{c_i}) \right) =   \Oplus_{i=1}^{n}  \Encpk(\mathrm{f_i}\mathrm{c_i})$ = $ \Encpk(\sum_{i=1}^{n}\mathrm{f_i}\mathrm{c_i})$, 

since $\sum_{i=1}^{n}\mathrm{f_i}\mathrm{c_i} = \D$.
\item  Sample  $r \la \mathbb{Z}_N$, and compute 
$r\oplus\Encpk(\D) = \Encpk(r+\D)$.
\item Send the following  value to \carol:
\begin{equation*}
\Encpk(r+\D)
 \end{equation*}
\end{compactenum}

\noindent{\bf{Round 3}.} \\
\noindent \carol performs the following computations.
\begin{compactenum}
\item Decrypt $\Encpk(r+\D)$ using $\sk$.
\item Compute five values:\\  

 $\frac{(r+\D)^2 }{(\B+\D)(\A+\C)}$, $\frac{(r+\D) }{(\B+\D)(\A+\C)}$, $\frac{(r+\D) }{(\A+\C)}$, $\frac{1}{(\B+\D)(\A+\C)}$ and $\frac{1}{\A+\C}$,\\

and encrypt them, obtaining:\\
 
 $\Encpk\left(\frac{(r+\D)^2 }{(\B+\D)(\A+\C)}\right)$, $\Encpk\left(\frac{(r+\D) }{(\B+\D)(\A+\C)}\right)$, $\Encpk\left(\frac{(r+\D) }{(\A+\C)}\right)$, $\Encpk\left(\frac{1}{(\B+\D)(\A+\C)}\right)$ and $\Encpk\left(\frac{1}{\A+\C}\right)$.\\
 
\item Send the five encrypted values to \felix.\\
\end{compactenum}

\noindent{\bf{Round 4}.} \\
\noindent \felix performs the following computations.
\begin{compactenum}

\item Eliminate $r$ from the first and third encrypted values by computing
\begin{equation*}
\begin{aligned}
& \Encpk\left(\frac{(r+\D)^2 }{(\B+\D)(\A+\C)}\right)
\oplus \left( {r^{2}} \otimes \Encpk\left(\frac{1}{(\B+\D)(\A+\C)}\right) \right) \\
& \oplus \left( -2{r} \otimes \Encpk\left(\frac{(r+\D) }{(\B+\D)(\A+\C)}\right) \right)
 = \Encpk\left(\frac{\D^2}{(\B+\D)(\A+\C)}\right) 
\end{aligned}
\end{equation*}
and
\begin{equation*}
\Encpk\left(\frac{(r+\D) }{(\A+\C)}\right)
\oplus \left( {-r} \otimes \Encpk\left(\frac{1}{(\A+\C)}\right) \right) \\
 = \Encpk\left(\frac{\D}{(\A+\C)}\right).
\end{equation*}
\item Compute an encryption of $\chisqfc$ by computing:
\begin{equation*}
\begin{aligned}
& \left( \frac{n^3}{(\A+\B)(\C+\D)} \otimes \Encpk \left( \frac{\D^2}{(\B+\D)(\A+\C)} \right)  \right)\\
&\oplus \left(\frac{n(\C+\D)}{\A+\B} \otimes \Encpk\left(\frac{\B+\D}{\A+\C}\right)\right)
\oplus \left(\frac{-2n^2}{\A+\B} \otimes \Encpk \left(\frac{\D}{\A+\C}\right) \right),
\end{aligned}
\end{equation*}
where $\C+\D $ and $\A+\B $ are computed as
\begin{equation*}
\C+\D = \sum_{i=1}^{n} \mathrm{f_i},
\end{equation*}
and
\begin{equation*}
\A+\B = n - (\C+\D).
\end{equation*}
We see below that the above computation gives $\Encpk(\chisqfc)$. 
Since $\A\D -\B\C = (\A+\B+\C+\D)\D - (\B+\D)(\C+\D)$, \chisqfc can be decomposed as follows:
\begin{equation*}
\begin{aligned}
\chisqfc &= \frac{n(\A\D-\B\C)^2}{(\A+\C)(\A+\B)(\C+\D)(\B+\D)} \\
&= \frac{n^3}{(\A+\B)(\C+\D)} \frac{\D^2}{(\B+\D)(\A+\C)} + \frac{n(\C+\D)}{(\A+\B)} \frac{(\B+\D)}{(\A+\C)} - \frac{2n^2}{(\A+\B)} \frac{\D}{(\A+\C)}. \\
\end{aligned}
\end{equation*}
\item Send the following value  to \carol:
\begin{equation*}
\Encpk(\chisqfc).
 \end{equation*}
\end{compactenum}

\noindent{\bf{Local computation}.} \\
\carol decrypts $\Encpk(\chisqfc)$ to obtain \chisqfc.

\section{Related Work}

There has been extensive research on privacy-preserving data mining (PPDM), which aims at completing data mining tasks on a union of several private datasets, each owned by a different party. The goal of PPDM can be achieved by either adding noise and perturbations~\cite{agrawal2000privacy,dwork2008differential} or using cryptographic tools. This paper falls into the latter category. 

General SMPC~\cite{yao1986generate,goldreich1987play,goldreich1998secure,henecka2010tasty,malkhi2004fairplay,ben2008fairplaymp} can be used to calculate any functions between multiple parties without revealing the input of each party. However, currently-known general SMPC protocols are computationally inefficient. Therefore, it is impractical to do large-scale multi-party feature selection using these protocols. Compared to general SMPC protocols, the protocol proposed in this paper is more efficient in handling feature selection. 

Recent studies have proposed several efficient SMPC protocols to accomplish different data mining tasks such as statistics computations~\cite{du2001privacy,canetti2001selective}, set intersections~\cite{freedman2004efficient,agrawal2003information}, classification~\cite{vaidya2005privacy,kikuchi2013privacy,vaidya2004privacy,wright2004privacy}, clustering~\cite{vaidya2003privacy}, and  regression~\cite{gascon2016secure}. However, to the best of our knowledge, not much research has been done in secure multi-party feature selection. As a commonly-used pre-processing technique, feature selection can be used in conjunction with many of the previously mentioned SMPC data mining protocols or as a metric to estimate data quality for classification tasks. 

There are many feature selection methods. \cite{du2001privacy} proposes an algorithm for privacy-preserving calculation of Pearson correlation coefficients among distributed parties. However, different from our approach, they use perturbation techniques to achieve privacy protection. \cite{banerjee2011privacy} proposes a secure multi-party feature selection protocol using virtual dimensionality reduction, but their protocol requires users to exchange unencrypted intermediate results such as the dot product of two attribute vectors. Our protocol achieves a stronger privacy protection: each participating party only learns the $\chi^2$ coefficient between the two attributes, and no intermediate results are leaked. 
\section{Conclusion}
Data trading will become more and more important as devices generate more and more data. In this work, we initiate a study on how to securely evaluate the value of trading data without requiring a trusted third party, by considering the specific case of data classification tasks. We present a secure four-round protocol that computes the value of the data to be traded without revealing the data to the potential acquirer. 

We employed additive homomorphic encryption as a core building block to compute the \chisq-statistic in a privacy-preserving manner.

% !TEX root = main.tex

\section{Acknowledgement}

The authors acknowledge and express appreciation for partial funding for this work from the U.S. Department of Transportation Federal Highway Administration Exploratory Advanced Research Program, grant ID DTFH6115H00006, and the support of FHWA Program Manager Dr. Ana Maria Eigen.

\bibliographystyle{splncs04}
\bibliography{references}

\end{document}